\documentclass{sig-alternate}

\usepackage[latin1]{inputenc}
\usepackage{makeidx}  
\usepackage{amsmath,amssymb,amsfonts,latexsym,epsfig} 
\usepackage{psfrag, xspace}
\usepackage{hyperref}
\usepackage{framed}
\usepackage{bm}
\usepackage{enumerate}
\usepackage{algorithmic} 
\usepackage[linesnumbered,vlined,boxruled,boxed,algo2e]{algorithm2e}

\usepackage{notations}

\usepackage{hyperref}
\newtheorem{theorem}{Theorem}
\newtheorem{lemma}[theorem]{Lemma}

\newtheorem{corollary}[theorem]{Corollary}

\newcommand{\x}{\ensuremath{\mathbf{x}}\xspace}

\def\wt#1{\widetilde #1}

\def\normi#1{\norm{#1}_{\infty}}

\usepackage[paper=letterpaper,width=7in,height=9.25in,centering]{geometry}

\begin{document}


\title{Accelerated Approximation of the Complex  Roots of a Univariate Polynomial\\
\textit{Extended Abstract} }

\numberofauthors{2}
\author{
  \alignauthor Victor Y.~Pan \\
  \affaddr{Depts. of Mathematics and Computer Science \\
    Lehman College and Graduate Center \\
    of the City University of New York \\
    Bronx, NY 10468 USA} 
%
%
%
  \email{victor.pan@lehman.cuny.edu}
  \url{http://comet.lehman.cuny.edu/vpan/} 
  \alignauthor Elias P.~Tsigaridas \\
  \affaddr{ 
    INRIA, Paris-Rocquencourt~Center, PolSys~Project \\
    UPMC, Univ Paris 06, LIP6 \\
    CNRS, UMR 7606, LIP6 \\
    Paris, France
    }
  \email{elias.tsigaridas@inria.fr}
}

\maketitle

 





 Highly efficient and even nearly optimal algorithms
have been
 developed for the classical problem of univariate polynomial root-finding (see, e.g., \cite{P95}, \cite{P02},
\cite{MNP13}, and the bibliography therein), but this is still
 an area of active research.
By combining some 
 powerful techniques developed in this area we
devise new nearly optimal 
algorithms, whose  substantial merit
is their simplicity, important for the implementation.
 
\bigskip 

We first recall the basic concept of the {\em isolation ratio},
central also for \cite{P95}, \cite{P02}.
Assume a real or complex polynomial 
$p=p(x)=\sum^{d}_{i=0}p_ix^i=p_n\prod^d_{j=1}(x-z_j),~~~ p_d\ne 0,$
of degree $d$, an annulus $A(X,R,r)=\{x: r\le |x-X| \le R \}$
 on the complex plane
 with a center $X$
and
the radii $r$ and $R$ of the boundary circles. Then the 
 internal
disc $D(X, r)=\{x: |x-X|\le r\}$ is  $R/r$-{\em isolated}
and $R/r$ is its {\em isolation ratio} 
if  the polynomial $p$ has no roots 
in the annulus.  
Next we reproduce \cite[Corollary~4.5]{renegar87}. It shows that
 Newton's iteration converges quadratically
 to a  single simple root of $p$ if is initiated at the center of a $5d^2$-isolated disc
that contains just this root.
\begin{theorem}\label{thren}
  Suppose both discs $D(c, r)$ and $D(c, r/s)$ for $s\ge 5d^2$
  contain a single simple root 
  $\alpha$ of 
  a polynomial $p=p(x)$ of degree $d$.
  Then  Newton's
  iteration 
  \begin{equation}\label{eqnewt}
    x_{k+1}=x_k-p(x_k)/p'(x_k),  k=0,1,\dots
  \end{equation}
  converges quadratically to the root $\alpha$
  right from the start provided $x_0=c$.
\end{theorem}

Now suppose that we are given a disc with a single simple zero of $p$
having an isolation ratio 
$1+\eta$ for a fixed  constant $\eta>0$.
Can we increase the ratio to $5d^2$?
Yes, we just need to apply a technique already used in 
\cite{schoenhage82} for the computation of the power sums of the roots
lying inside such a disc. In our case this is a single root,
the power sum is the root itself, 
and we just need its approximation $c$ within an error at most $\Delta$
such that $r\eta/\Delta \ge 5d^2$. Indeed in this case
$\Delta\le 0.2r\eta/d^2$, and so the disc $D(c,\Delta)$ is $5d^2$-isolated. 

We can shift and scale the variable $x$, and so 
wlog we assume dealing with 
a $(1+t)^2$-isolated disc $D(0,r)$ for $r=1/(1+t)$ for a fixed $t>0$,
and with polynomial $p$ having a single simple root $z_1$ in this disc.
Recall
the Laurent expansion,

\begin{equation}\label{equ7.12.3}
  \begin{aligned}
    \frac{p'(x)}{p(x)} & =\sum_{j=1}^d\frac{1}{x-z_j} 
    = -\sum_{k=1}^\infty S_k x^{k-1}+\sum_{k=0}^\infty s_kx^{-k-1}  \\
    & = \sum_{h=-\infty}^{\infty}c_hx^h.
\end{aligned}
\end{equation}
Here $|x|=1$,
$s_0=1,~~s_k=z_{1}^k,~~S_k=
\sum_{i=2}^d~z_{i}^{-k},~k=1,2,\ldots$
Consequently $s_k=z_1^k$,
whereas $S_k$ is the $k$th power sum of the zeros of
the reverse polynomial $p_{\textrm{rev}}(x)$ that lie in the disc
$D(0,r)$. 
The leftmost equation of
(\ref{equ7.12.3}) is verified by the differentiation of
$p(x)=p_n\prod_{j=1}^d(x-z_j)$. The middle equation 
is implied by the  decompositions 
$\frac{1}{x-z_{1}}
=~\frac{1}{x}\sum_{h=0}^\infty\left(\frac{z_{1}}{x}\right)^h
$ and 
$\frac{1}{x-z_{i}}
=~-\frac{1}{z_{i}}\sum_{h=0}^\infty\left(\frac{x}{z_{i}}\right)^h
\mbox{
for }i>1$, provided $|x|=1$ for all $i$.

We cover the case of any 
positive integer $k$, although we only need the case where $k=1$.
For a fixed positive integer
$q$ we
compute the  approximations $s_k^*\approx s_k$ as follows,
\begin{equation}\label{equ7.12.2}
s_k^*=\frac{1}{q}\sum_{j=0}^{q-1}\omega^{j(k+1)}p(\omega^j)/p'(\omega^j),
~~k=1,2,\ldots,q-1.
\end{equation}
Here $\omega=\omega_q=\exp(2\pi\sqrt {-1}/q)$ is a primitive $q$th
root of unity. 
Then the evaluation of 
the polynomial $p(x)$ at the $q$th roots of unity amounts to the same task for
a polynomial $p_q(x)$ of degree at most $q-1$ with the coefficients 
$p_{q,i}=\sum_{j=0}^lp_{i+jq}$ for $l=\lfloor d/q\rfloor$
obtained by means of less than $d$ additions of the coefficients of $p$.

Having computed the polynomial $p_q(x)$
we reduce the evaluation of all the desired approximations $s_k^*$ 
for $k=1,\ldots,q-1$
essentially to
performing three DFTs,\index{DFT}
 each on $q$ points, that is to a total of
$\OO(q\log (q))$ ops.\index{op}
 Namely, we apply two DFTs\index{DFT} 
 to compute
$p(\omega^i)$ and $p'(\omega^i)$ for $i=0,1,\ldots,q-1$ and a single
DFT\index{DFT} to multiply the DFT matrix
$\Omega=[\omega^{hi}]_{h,i=0}^{q-1}$ by the vector
${\bf v}=[p(\omega^i)/p'(\omega^i]_{i=0}^{q-1}$.

Let us estimate the approximation errors. Equations (\ref{equ7.12.3}) and
(\ref{equ7.12.2}) imply that
$$s_k^*=\sum_{l=-\infty}^{+\infty}c_{-k-1+lq}.$$ Moreover,
(\ref{equ7.12.3}) for $h=-k-1,k\ge 1$ implies that $s_k=c_{-k-1}$,
whereas (\ref{equ7.12.3}) for $h=k-1, k\ge 1$ implies that
$S_k=-c_{k-1}$. Consequently
$$s_k^*-s_k=\sum_{l=1}^{\infty}(c_{lq-k-1}+c_{-lq-k-1}).$$ We assumed
in (\ref{equ7.12.2}) that $0<k<q-1$. It follows that
$c_{-lq-k-1}=s_{lq+k}$ and $c_{lq-k-1}=-S_{lq-k}$ for
$l=1,2,\ldots$, and we obtain
\begin{equation}\label{equ7.12.5}
s_k^*-s_k=\sum_{l=1}^\infty(s_{lq+k}-S_{lq-k}).
\end{equation}
On the other hand $|s_h|\le
z^h,~~|S_h|\le(d-1)z^h$, $h=1,2,\ldots$ where
$z=\max_{1\le j\le d}\min(|z_j|,1/|z_j|)$,
and so 
$z\le \frac{1}{1+t}$ in our case.
Substitute these bounds into (\ref{equ7.12.5}) and obtain
$|s_k^*-s_k|\le(z^{q+k}+(d-1)z^{q-k})/(1-z^q)$.
Therefore it is sufficient to  choose $q$ of order $\log (d)$ to 
decrease the error of the approximation to the root 
 $z_1$ by a factor of $gd^h$ for any pair of constants $g$ and
$h$, and so we can ensure the desired error bound $\Delta$.  
To support this computation we only need
 less than $d$ additions, followed by $\OO(\log (d))$
evaluations of the polynomial $p_q(x)$
of degree $q-1$ 
at the $l$th roots of unity for $l=O(\log (d))$.
This involves $\OO(\log (d)\log (\log (d)))$ ops overall.
(Here and hereafter ``ops" stand for ``arithmetic operations".)
Summarizing we obtain the following estimates.


\begin{theorem}\label{thisol}
Suppose the unit disc $D(0,r)=\{x:~|x|\le 1\}$ is $(1+\eta)^2$-isolated
for $(1+\eta)r=1$ and a fixed $\eta>0$ and
 contains a single simple root $z$ of a polynomial $p=p(x)$ of a degree $d$. Then 
it is sufficient to apply less than $d$ additions
and $\OO(\log (d)\log (\log (d)))$ other  ops
to compute a $5d^2$-isolated subdisc of $D(0,r)$ containing this root.
\end{theorem}

Combine Thm.~\ref{thren} and \ref{thisol} and obtain the following result.


\begin{corollary}\label{corref}
Under the assumptions of Theorem \ref{thisol}
we can approximate the root $z$ of the polynomial $p(x)$
within a fixed positive error bound $\epsilon<1$ by using 
$\OO(\log (d)\log (\log (d))+d\log (\log (1/\epsilon)))$ ops.
\end{corollary}


\begin{corollary}\label{corref1}
Suppose that we are given $d$ discs, each containing a single simple root 
of a polynomial $p=p(x)$ of degree $d$ and each being $(1+\eta)^2$-isolated 
for a fixed $\eta>0$.
Then we can approximate all $d$ roots  of this polynomial 
within a fixed positive error bound $\epsilon<1$ by using 
$\OO(d\log^2 (d)(1+\log (\log (1/\epsilon))))$ ops.
\end{corollary} 


\begin{proof}
 Apply the same algorithm that supports Corollary~\ref{corref}
concurrently in all $d$ given discs, but instead of the
$q$th roots of unity use $q$ equally spaced points at the boundary 
circle of each input disc (that is $dq=O(d\log d)$ points overall)
and instead of FFT apply 
the Moenck--Bo\-ro\-din algorithm for multipoint polynomial evaluation.
Also use it at the stage of performing concurrent Newton's iteration
 initialized at the centers  of the $5d^2$-isolated subdiscs of the 
$d$ input discs, each subdisc computed by the
 algorithm that supports Theorem~\ref{thisol}.
Here we work with the $d$th degree  polynomial $p$ rather than with
the $q$th degree polynomials $p_q$
because to support transition to polynomials $p_q$ of the degree $q$
 for $d$ discs we would need to perform $d$ shifts and scalings
of the variable $x$. Instead we employ the Moenck--Bo\-ro\-din algorithm,
which still enables us to obtain a nearly optimal root-refiner.
Technically, in a relatively minor change of our algorithm, we replace the   
matrix $\Omega=[\omega^{j(k+1)}]_{j,k}$ in (\ref{equ7.12.2})  
by the matrix $[c+\omega^{j(k+1)}]_{j,k}=c[1]_{j,k}+\Omega$ where $c$ is 
 invariant in $j$ and $k$.
The multiplication of the new matrix by a vector ${\bf v}$ is still reduced to 
multiplication of  the matrix $\Omega$ by a vector ${\bf v}$
with the additional $3d$ ops for computing the vector $c[1]_{j,k}{\bf v}$
and adding it to the vector $\Omega {\bf v}$.
\end{proof}

The Moenck--Borodin algorithm uses nearly linear arithmetic time, and
\cite{K98}
proved that this algorithm supports  multipoint 
polynomial evaluation at a low Boolean cost as well 
(see also
 \cite{vdH08}, \cite{PT13}, \cite{KS13},
\cite{PTa}, \cite{Pa}, \cite{Pb}). 
Consequently {\em our algorithm supporting Corollary \ref{corref1}
can be extended to support a nearly optimal 
Boolean cost bound for refining all simple isolated roots of a polynomial}.

We can immediately relax the assumption that the roots are simple 
because our proof of Theorem \ref{thisol} applies to a multiple root as well.
Furthermore deduce from the Lucas theorem that the isolation ratio
of the basic discs in our algorithms does not decrease when we shift from a polynomial to its derivative and 
higher order derivatives. Therefore we can just apply
Newton's iteration to the derivative or to a higher order derivative to
approximate a double or multiple root, respectively.  

\subsection*{Boolean cost bounds}

Hereafter \sOB 
denotes the
bit or Boolean complexity ignoring logarithmic factors.
To estimate it we apply some
results from \cite{PT13}--\cite{PT14},
which hold 
in the general case
where the coefficients of the polynomials are known up to an arbitrary
precision. In our case the input polynomial is known exactly; 
the parameter $\lambda$ 
to be specified
in the sequel could be considered as  the working precision.

Let $p$ be given as 
a $\lambda$-approximation, ie $\lg\normi{p -\wt{p}} \leq -\lambda$.
We compute $p_q$ 
by using $d$ additions.
This produces
a polynomial such that
$\lg \normi{p_q} \leq \tau + \lg{d}$, and
$\lg \normi{p_q - \wt{p}_q} \leq -\lambda + \tau \lg{d} + 1/2\lg^2{d} + 1/2\lg{d}
=  \OO(-\lambda + \tau\lg{d} + \lg^2{d})$.

Similar bounds hold for $p_q'$, ie
$\lg \normi{p'_q} \leq \tau + 2\lg{d}$, and
$\lg \normi{p'_q - \wt{p'}_q} \leq  -\lambda + \tau \lg{d} + 3/2\lg^2{d} + 1/2\lg{d}
=  \OO(-\lambda + \tau\lg{d} + \lg^2{d})$.

Recall that
$\abs{p'_q(\omega^{i})} \leq  \tau + 2\lg{d} + \lg\lg{d} + 2$
and  
$\abs{p'_q(\omega^{i}) - \wt{{p'_q(\omega^{i})}}} \leq 
-\lambda + \tau\lg(2d) + 3/2\lg^2{d} + 5/2\lg{d} + \lg{\lg{d}} + 5$
for all $i$,
\cite[Lemma~16]{PTa}, 
and similar bounds hold for $p_q(\omega^{i})$.

The divisions
 $p_q(\omega^{i})/p'_q(\omega^{i})$ 
output
complex numbers such that 
$\abs{p_q(\omega)/p'_q(\omega)} \leq  \tau + 2\lg{d} + \lg\lg{d} + 2$
  with
the logarithm of the error 
$\leq -\lambda + \tau\lg(4d) + 3/2\lg^2{d} + 9/2\lg{d} + 2\lg{\lg{d}} + 11$.

The final DFT produces numbers
such that the logarithms of their magnitudes 
are not greater than
$\tau + 2\lg{d} + 2\lg\lg{d} + 4$
and the logarithms of their approximation errors 
are at most
$ -\lambda + \tau\lg(8d) + 3/2\lg^2{d} + 13/2\lg{d} + 4\lg{\lg{d}} + 18$,
\cite[Lemma~16]{PTa}.

To achieve an error 
within  $2^{-\ell}$ in the final result, we 
perform all the computations with accuracy
$\lambda = \ell +  \tau\lg(8d) + 3/2\lg^2{d} + 13/2\lg{d} + 4\lg{\lg{d}} + 18$,
that is 
$\ell = \OO( \ell + \tau \lg{d} + \lg^{2}{d}) = \sO(\ell + \tau)$.

We perform $d$ additions 
at the cost $\OB(d \lambda)$
and perform
the rest of computations, that is the 3 DFTs,
at the cost
$\OB( \lg{d} \, \lg\lg{d} \, \mu(\lambda))$
or $\sOB(d (\ell + \tau))$ \cite[Lemma~16]{PTa}.

If the root that we want to refine is not in the unit 
disc,
 then 
we 
replace $\tau$ in 
our bounds with 
$d \tau$.

We apply a similar analysis 
from
\cite[Section~2.3]{PT13} 
to the Newton 
iteration
 (see also
\cite[Section~2.3]{PTa}) and 
arrive at  the same 
asymptotic 
bounds on the Boolean complexity. 
Only the overhead constants
change 
because now we perform computations with complex numbers.

The 
overall complexity is $\sOB(d^2 \tau + d \ell)$
and the working precision is $\OO(d\tau + \ell)$.

Here we assume the  exact input, that is 
assume the coefficients of the input
polynomials known up to arbitrary precision. 
For the refinement of 
the 
root
up to precision of $L$ bits, we arrive at
an algorithm with the  complexity
in 
$\sOB(d^2 \tau + d L)$.

If we are interested in refining all complex roots, we cannot work
anymore with the polynomial $p_q$ of degree $q = \OO(\lg{d})$ unless
we add the cost of $d$ shifts of the initial approximations to the
origin.  Instead we rely on fast algorithms for multipoint
evaluation.  Initially we evaluate the polynomial $p$ of degree $d$ at
$\OO(d \lg{d})$ points, and we assume that $\lg{\normi{p}} \leq \tau$.
These $d$ points approximate the roots of $p$, and so their magnitude
is at most $\leq 2^{\tau}$.

We use the following result of \cite[Lemma 21]{PT14}.
Similar bounds appear in \cite{K98,KS13,vdH08}.
\begin{lemma}[Modular representation]
  \label{lem:rem-m-polys}
  Assume $m+1$ polynomials,
 $F \in \CC[x]$ of degree $2 m n$ 
  and $P_j \in \CC[x]$ of degree $n$, for $j=1,\dots, m$
  such that $\normi{F} \leq 2^{\tau_1}$ and all roots of 
  the polynomials $P_j$
  for all $j$ have magnitude of at most
  $2^{\rho}$. 
  Furthermore assume $\lambda$-approximations of $F$ by  $\wt{F}$ and   
  of $P_j$ by $\wt{P}_j$
  such that 
  $\normi{F - \wt{F}} \leq 2^{-\lambda}$
  and 
  $\normi{P_j - \wt{P}_j} \leq 2^{-\lambda}$.
  Let  $ \ell = \lambda - \OO(\tau_1\lg{m} + m \, n \, \rho)$.
  Then we can compute 
  an $\ell$-approximations $\wt{F}_j$ of $F_j = F \mod P_j$
  for $j=1,\dots, m$
  such that $\normi{F_j - \wt{F}_j} \leq 2^{-\ell}$
  in 
  $\sOB(m \, n \,(\ell + \tau_1 + m\,n\,\rho))$.
\end{lemma}

Using this lemma we bound the overall complexity of multipoint
evaluation by $\sOB(d(L + d \tau))$.  The same bounds holds at the
stage where we perform Newton's iteration.  We need to apply Newton's
operator $\sO(1)$ for each root.  Each application of the operators
consists of two polynomial evaluations.  We perform the evaluations
simultaneously and apply Lemma \ref{lem:rem-m-polys} to bound the
complexity.  On similar estimates for the refinement of the real roots
see \cite{PTa}.

\subsection*{Extensions}

The algorithm of \cite{MSW13}
computes at nearly optimal cost $64d$-isolated initial discs for all $d$ roots 
of a polynomial $p(x)$. By combining this algorithm with ours   
we    obtain a distinct alternative algorithm, which like the one of \cite{P95},
\cite{P02}, 
{\em supports  the record nearly optimal bounds on the Boolean complexity of 
the approximation of all complex polynomial roots},
but has the advantage of allowing substantially simpler implementation. 

Finally the same algorithm of \cite{schoenhage82} approximate the power sums
of any number $m$ of roots (forming, e.g., a single cluster or a number of clusters) 
in an isolated disc. The algorithm runs at about the same cost, already stated and depending 
just on the isolation ratio. Having the power sums available
we can readily compute the coefficients of the factor $f$ of $p$ of degree $d_f$,
whose roots are exactly the roots of $p$ in this disc: this is a
numerically stable algorithm using $\OO(m\log(m))$ ops  (cf. \cite[pages 34--35]{BP94}).


{
  \scriptsize
  \textbf{Acknowledgments.}
  VP is supported by NSF Grant CCF--1116736. 
  ET is partially supported by 
  GeoLMI 
  (ANR 2011 BS03 011 06), 
  HPAC (ANR ANR-11-BS02-013)
  and an
  FP7 Marie Curie Career Integration Grant.  
}

{ 
  \scriptsize


\begin{thebibliography}{10}







\bibitem{BP94}
D.~Bini and V.~Pan.
\newblock {\em Polynomial and Matrix Computations}, volume 1: Fundamental
  Algorithms.
\newblock Birkh\"{a}user, Boston, 1994.

\bibitem{K98}
P. Kirrinnis,
Polynomial Factorization and Partial Fraction Decomposition by
  Simultaneous Newton's Iteration,
\textit{J. of Complexity} \textbf{14}, 378--444  (1998).


\bibitem{KS13}
A. Kobel and M. Sagraloff,
Fast Approximate Polynomial Multipoint Evaluation and Applications,
arXiv:1304.8069v1 [cs.NA] 30 April 2013. 


\bibitem{MNP13}
J.~M. McNamee. and V.~Y. Pan,
{\em Numerical Methods for Roots of Polynomials},  
 Part 2 (XXII + 718 pages), Elsevier (2013).


\bibitem{MSW13}
K. Mehlhorn, M., Sagraloff, P. Wang,
From Approximate Factorization to Root Isolation with Application to Cylindrical Algebraic Decomposition, 
in {\em Proc. Intl. Symp. on Symbolic and Algebraic Computations                          
(ISSAC)}, Boston,
(M. Kauers, editor), 283--290, ACM   
Press, New York (2013).


\bibitem{P95}
 V. Y. Pan, 
Optimal (up to Polylog Factors) Sequential and
Parallel Algorithms
  for Approximating Complex Polynomial Zeros,
\textit{Proc. 27th Ann. ACM Symp. on Theory of Computing (STOC '95)}, 
ACM Press, New York, 741--750 (1995).




\bibitem{P02}
V. Y. Pan,
Univariate Polynomials: Nearly Optimal Algorithms for Factorization and
 Rootfinding,
{\em Journal of Symbolic Computations}, {33}, { 5}, 701--733, 2002.




\bibitem{Pa}
V. Y. Pan, 
Transformations of Matrix Structures Work Again, 
accepted by {\em Linear Algebra and Its Applications}, 
available at arXiv:1311.3729v1 [math.NA] 15 Nov 2013.


\bibitem{Pb}
V. Y. Pan,
Fast Approximation Algorithms for Computations with Cauchy Matrices and Extensions,
Tech. Report TR-2014005,
\textit{PhD Program in Comp. Sci.}, \textit{Graduate Center, CUNY}, 2014

Available at http://tr.cs.gc.cuny.edu/tr/techreport.php?id=469
 
Proc. version in Proceedings of CSR 2014 (E.A. Hirsch et al. (Eds.)), LNCS 8476, pp. 287-300, 2014
(Springer International Publishing, Switzerland 2014).



\bibitem{PT13}
V. Y. Pan and E. P. Tsigaridas, in {\em Proc. International Symp. on Symbolic and Algebraic Computations
(ISSAC 2013)}, Boston, Massachusetts, June 2013, (M. Kauers, editor), 299--306, ACM   
Press, New York (2013).       


\bibitem{PTa}
V. Y. Pan and E. P. Tsigaridas, 
Nearly Optimal Refinement of Real Roots of a Univariate Polynomial,
Tech. Report, INRIA (2013).
url: {http://hal.inria.fr/hal-00960896},


\bibitem{PT14}
V. Y. Pan and E. P. Tsigaridas, 
Nearly Optimal Computations with Structured Matrices,
Tech. Report, INRIA (2014) (Submitted).


\bibitem{renegar87}
J. Renegar,
On the worst-case arithmetic complexity of approximating zeros of
polynomials, \textit{J. of Complexity} \textbf{3,~2}, 90--113 (1987).


\bibitem{vdH08}
 J. van der Hoeven,
Fast composition of numeric power series,
Tech. Rep. 2008-09,  Universit{\'e} Paris-Sud, Orsay, France,
  2008.


\bibitem{schoenhage82}
 A.~Sch{\"o}nhage,
The Fundamental Theorem of Algebra in Terms of
                   Computational Complexity,
  manuscript, Univ. of T\"{u}bingen, Germany, 1982,
  URL:~http://www.iai.uni-bonn.de/\textasciitilde

\end{thebibliography}
  
}

\end{document}